\documentclass[a4paper,UKenglish,cleveref,autoref]{lipics-v2021}
\hideLIPIcs
\nolinenumbers
\usepackage{booktabs,tabularx,array}
\usepackage{tikz,graphicx}
\usetikzlibrary{arrows.meta,positioning}
\newcommand{\Npos}{\mathbb N_{>0}}
\newcommand{\dist}{\operatorname{dist}}
\newcommand{\pc}{\operatorname{pc}}

\newcommand{\NPC}{\textnormal{NP-complete}}
\newcommand{\LT}{\textnormal{Linear time}}
\newcolumntype{Y}{>{\raggedright\arraybackslash}X}

\IfFileExists{plainurl.bst}{\bibliographystyle{plainurl}}{\bibliographystyle{plain}}

\title{The Complexity of Hamiltonian Cycles and Paths with Prescribed Walk Lengths}
  \titlerunning{Hamiltonian Cycles and Paths with Prescribed Walk Lengths}

\author{Hitendra Kumar}{Indian Institute of Technology Dharwad, India}{hitendra.kumar1729@gmail.com}{}{}
\author{Mathieu Mari}{LIRMM, University of Montpellier, Montpellier, France. \and \url{https://www.lirmm.fr/mathieu-mari/index.html}}{mathieu.mari@lirmm.fr}{https://orcid.org/0000-0001-8074-0241}{}
\author{R. B. Sandeep}{Indian Institute of Technology Dharwad, India \and \url{https://sites.google.com/site/homepagesandeeprb}}{sandeeprb@iitdh.ac.in}{https://orcid.org/0000-0003-4383-1819}{R. B. Sandeep thanks the IWOCA~2026 organizing team and gratefully
acknowledges financial support from LIRMM, Montpellier, France,
through the IWOCA~2026 Extended Stay Support Scheme for his research
visit to LIRMM.}

  \authorrunning{H. Kumar, M. Mari, R.B. Sandeep}
  \Copyright{Anonymous}
  \ccsdesc[500]{Theory of computation~Problems, reductions and completeness}
  \ccsdesc[500]{Mathematics of computing~Graph algorithms}
  \keywords{Hamiltonian cycle, Hamiltonian path, walk length, graph powers,
    complexity dichotomy, directed graphs}

\begin{document}
\maketitle

\begin{abstract}
Fix a set $S$ of positive integers. An $S$-Hamiltonian cycle is a cyclic
ordering of all vertices of a graph such that each consecutive pair is
joined by a walk whose length belongs to $S$. We study the complexity
of deciding whether an input graph admits such an ordering, together
with the corresponding path variants with prescribed or unrestricted
endpoints. We complete the complexity classification for undirected
graphs. In particular, we prove NP-completeness for every finite
nonempty set of odd lengths, resolving an open problem posed by
Amarilli, Lombardo, and Monet (WG 2026). 
For digraphs, all three problems are NP-complete for every
fixed finite nonempty $S$ and admit linear-time algorithms when $S$
is cofinite. For other fixed infinite sets $S$, we establish
NP-complete and undecidable cases, with undecidability occurring
even when $S$ is recursive.
\end{abstract}
\newpage

\section{Introduction}\label{sec:introduction}

The \emph{Hamiltonian Cycle} and \emph{Hamiltonian Path} problems ask
whether a given graph contains a cycle or a path, respectively, that
visits every vertex exactly once. These are fundamental problems in
graph theory and theoretical computer science~\cite{Gould2003},
closely connected to routing and sequencing applications through the
travelling salesman problem~\cite{ApplegateEtAl2006}. As classical
NP-complete problems, they also play a central role in the study of
computational complexity and the design of
algorithms~\cite{Karp1972,GareyJohnson1979}. In both problems,
consecutive vertices must be adjacent. A natural generalisation,
studied by Amarilli, Lombardo, and Monet~\cite{ALM}, relaxes this
requirement by allowing consecutive vertices to be joined by a walk
whose length belongs to a fixed set of positive integers. Every vertex
must still appear exactly once in the cyclic or linear ordering, but
the walks witnessing these transitions may repeat vertices and pass
through vertices appearing elsewhere in the ordering.

Formally, fix $S\subseteq\Npos$. An \emph{$S$-Hamiltonian cycle} is a
cyclic ordering of all vertices in which every consecutive pair is joined
by a walk of length in $S$. We also consider \emph{$S$-Hamiltonian paths},
both with prescribed endpoints and with unrestricted endpoints. For a
digraph, every witnessing walk must respect the arc directions. Throughout
the paper, $S$ is fixed rather than supplied as part of the input. Crucially,
these are conditions on \emph{walk lengths}, not on shortest-path distances.
For example, an edge of an undirected graph can be traversed back and
forth to give a walk of any positive odd length between its endpoints.

More generally, a walk of positive length $t$ can be extended to length
$t+2i$ for every integer $i\geq0$ by inserting back-and-forth traversals
of an edge. Hence the undirected problems for $S$ are identical to those
for its \emph{normalisation}, the parity-preserving downward closure
\[
 \operatorname{cl}_2(S)=
 \{r\in\Npos:\exists\ell\in S,\ r\leq\ell,\ r\equiv\ell\pmod2\}.
\]
We call $S$ \emph{normalised} when $S=\operatorname{cl}_2(S)$.
Such a set is a union of an odd prefix and an even prefix: write
$O_k=\{1,3,\ldots,k\}$ for odd $k$, $E_h=\{2,4,\ldots,h\}$ for even $h$,
and $O_\infty,E_\infty$ for the entire positive odd and even classes.
Either prefix may also be empty.

\paragraph*{Context and previous work.}
Amarilli, Lombardo, and Monet introduced the $S$-Hamiltonian variant of
Hamiltonicity and studied its complexity on undirected graphs~\cite{ALM,ALMfull}.
For normalised sets, their cycle classification gives NP-completeness
for $\{1\}$, $\{1,2\}$, and $\{2\}$, and NP-hardness under Cook reductions
for $\{2,4\}$. The remaining classified cases are linear-time solvable:
even prefixes containing $\{2,4,6\}$, mixed sets containing
$\{1,2,3\}$ or $\{1,2,4\}$, and $O_\infty$ (with the empty set trivial).
The first two classical cases follow from ordinary Hamiltonicity and
graph squares~\cite{Karp1972,Underground}; the tractability of sets
containing $\{1,2,3\}$ follows from the Hamiltonian-connectedness of
graph cubes~\cite{Karaganis}.

The finite odd-only prefixes $O_k$ with $k\geq3$ were the sole remaining
family in their cycle classification. For paths, they additionally left
open $\{2,4\}$ with unrestricted endpoints and $\{2,4,6\}$ with prescribed
endpoints. Their $\{2,4\}$ hardness proofs for cycles and prescribed
endpoints used Cook reductions, leaving many-one hardness unresolved.
They also proposed studying directed graphs as an open direction.

\paragraph*{Our contributions.}
We resolve all these undirected open cases. Our main result,
Theorem~\ref{thm:odd-hard}, proves NP-completeness for every finite odd
prefix, completing the cycle classification. For $\{2,4\}$, we establish
NP-completeness under many-one reductions for all three variants.
For prescribed-endpoint paths with $\{2,4,6\}$, we give a linear-time algorithm.
Table~\ref{tab:undirected} gives the complete classification, with our
new results in bold.

\begin{table}[!ht]
\centering\scriptsize
\setlength{\tabcolsep}{3pt}
\begin{tabularx}{\textwidth}{@{}YYYY@{}}
\toprule
Normalised set & Cycle & Prescribed-endpoint path & Unrestricted path\\
\midrule
$O_1=\{1\}$
& \NPC\ \cite{Karp1972}
& \NPC\ \cite{IPS82}
& \NPC\ \cite{IPS82}\\

$O_k$, finite odd $k\geq3$
& \textbf{NP-complete\ (Thm.~\ref{thm:odd-hard})}
& \textbf{NP-complete\ (Thm.~\ref{thm:odd-hard})}
& \textbf{NP-complete\ (Thm.~\ref{thm:odd-hard})}\\

$\{2\}$
& \NPC\ \cite{ALM,ALMfull}
& \NPC\ \cite{ALM,ALMfull}
& \NPC\ \cite{ALM,ALMfull}\\

$\{2,4\}$
& \textbf{NP-complete\ (Thm.~\ref{thm:twofour})}
& \textbf{NP-complete\ (Thm.~\ref{thm:twofour})}
& \textbf{NP-complete\ (Thm.~\ref{thm:twofour})}\\

$\{1,2\}$
& \NPC\ \cite{Underground}
& \NPC\ \cite{ALMfull}
& \NPC\ \cite{ALMfull}\\

$E_h$, $h\geq6$ or $h=\infty$
& \LT\ \cite{ALM,ALMfull}
& \textbf{Linear time\ (Thm.~\ref{thm:even-easy})}
& \LT\ \cite{ALM,ALMfull}\\

$\{1\}\cup E_h$, $h\geq4$ or $h=\infty$
& Trivial\ \cite{ALM,ALMfull}
& \LT\ \cite{ALM,ALMfull}
& Trivial\ \cite{ALMfull}\\

$\{1,2,3\}\subseteq\operatorname{cl}_2(S)$
& Trivial\ \cite{Karaganis}
& Trivial\ \cite{Karaganis}
& Trivial\ \cite{Karaganis}\\

$O_\infty$
& \LT\ \cite{ALM,ALMfull}
& \LT\ \cite{ALM,ALMfull}$^\dagger$
& \LT\ \cite{ALM,ALMfull}$^\dagger$\\

$\varnothing$
& Trivial
& Trivial
& Trivial\\
\bottomrule
\end{tabularx}
\caption{The complete undirected classification.
Entries marked $\dagger$ follow by adapting the proof of
Theorem~7.1 in~\cite{ALMfull}.
NP-completeness uses many-one reductions.
``Trivial'' means a constant answer on connected inputs of the
relevant minimum order (yes for the mixed rows, no for the empty set);
checking connectedness takes linear time.
For the empty set, the unrestricted path problem accepts exactly
the one-vertex inputs, whereas the prescribed-endpoint problem
has no yes-instances.}
\label{tab:undirected}
\end{table}

For directed graphs, we answer the finite-set question: every fixed
finite nonempty $S$ gives NP-completeness, even on oriented graphs.
For infinite sets, cofinite $S$ gives linear-time algorithms, whereas
we exhibit both a fixed infinite NP-complete case and a fixed recursive
infinite undecidable case, even on strongly connected oriented graphs.
Table~\ref{tab:directed} summarises these results; a complete
classification of infinite directed sets remains open.

\begin{table}[!ht]
\centering\scriptsize
\setlength{\tabcolsep}{3pt}
\begin{tabularx}{\textwidth}{@{}Yl@{}}
\toprule
Fixed set and input restriction & All three problems\\
\midrule
$S=\varnothing$ & Trivial\\
Finite nonempty $S$, oriented graphs & \textbf{NP-complete\ (Thm.~\ref{thm:directed-finite})}\\
Cofinite $S$ & \textbf{Linear time\ (Thm.~\ref{thm:cofinite})}\\
$S_0=\{2\}\cup O_\infty$, strongly connected bipartite oriented graphs & \textbf{NP-complete\ (Thm.~\ref{thm:infinite-hard})}\\
A recursive infinite $S^\star$, strongly connected oriented graphs & \textbf{Undecidable (Thm.~\ref{thm:undecidable})}\\
\bottomrule
\end{tabularx}
\caption{Directed results; bold entries are our contributions. ``Trivial'' means no on inputs of at least two
vertices; a single vertex is a path. NP-completeness uses many-one
reductions. These rows do not classify all infinite sets.}
\label{tab:directed}
\end{table}

\paragraph*{Organisation.}
Section~\ref{sec:preliminaries} establishes the common framework.
Sections~\ref{sec:undir-hard} and~\ref{sec:undir-easy} prove the
undirected results. Sections~\ref{sec:dir-finite} and~\ref{sec:dir-infinite}
treat finite and infinite directed length sets, respectively.

\section{Preliminaries}\label{sec:preliminaries}

We relate the three problems to ordinary Hamiltonicity in the
compatibility graph $G_S$, giving polynomial-time construction and
NP membership for fixed undirected sets and finite directed sets.
All graphs are finite. Undirected graphs are simple. Digraphs have no
parallel arcs and may have loops unless a restriction is stated; an
\emph{oriented graph} has neither loops nor antiparallel arcs. Write
$V(G)$ for the vertex set, and $E(G)$ or $A(G)$ for the edge or arc set.
A walk of length $\ell$ has $\ell$ edges or arcs, counted with multiplicity.
Its vertices need not be distinct. A path has no repeated vertices. Write $G[U]$ for the subgraph induced
by $U$, and $N_G(v),N_G[v]$ for the open and closed neighbourhoods of
$v$. In digraphs, $N_G^+(v),N_G^-(v)$ and $d_G^+(v),d_G^-(v)$ denote
out- and in-neighbourhoods and degrees. A digraph is bipartite when its
underlying undirected graph is bipartite. The line graph of an
undirected graph has one vertex for each edge, with adjacency when
the corresponding edges share an endpoint.
A graph or digraph is \emph{Hamiltonian-connected} if it has a Hamiltonian
path between every pair of distinct vertices (in the specified order for
a digraph). The power $G^r$ of an undirected graph joins distinct vertices
whose distance in $G$ is at most $r$.

Fix $S\subseteq\Npos$. The \emph{$S$-Hamiltonian Path} problem asks whether
$V(G)$ has an ordering $v_1,\ldots,v_n$ such that each $v_i$ is joined to
$v_{i+1}$ by a walk of length in $S$. The \emph{prescribed-endpoint}
variant additionally requires $v_1=s$ and $v_n=t$, for two given distinct
vertices. The \emph{$S$-Hamiltonian Cycle} problem requires $n\geq3$ and
also a permitted walk from $v_n$ to $v_1$. For digraphs, walks must respect
arc directions. Witnessing walks are independent and may repeat vertices
or pass through any vertex of the ordering.
A single vertex is a Hamiltonian path; empty inputs and other bounded-order
cases are handled directly. Unless stated otherwise, paths have at least
two vertices and cycles at least three. Directed cycles used to pad walks
may have length one or two.

All NP-hardness reductions below are polynomial-time many-one
(\emph{Karp}) reductions. Undecidability uses computable many-one
reductions. No representation of the fixed set $S$ is part of an input.

\begin{definition}[Compatibility graph or digraph]\label{def:compatibility}
For an undirected graph $G$, let $G_{S}$ be the simple graph on
$V(G)$ in which distinct $x,y$ are adjacent precisely when an $x$--$y$
walk in $G$ has a length in $S$. For a digraph $G$, define the loopless
digraph $G_S$ analogously, using directed walks from $x$ to $y$.
\end{definition}

Each $S$-Hamiltonian problem on $G$ is exactly its ordinary Hamiltonian
counterpart on $G_S$. For every fixed undirected $S$, this graph can be
constructed in polynomial time: breadth-first search in the parity lift
$V(G)\times\{0,1\}$, with edges $(u,i)(v,1-i)$ for $uv\in E(G)$, finds
the shortest walks of each parity; their lengths are compared with the
fixed prefixes of $\operatorname{cl}_2(S)$. For finite directed $S$,
exact-length reachability up to $\max S$ likewise constructs $G_S$ in
polynomial time (the empty set gives no arcs).
Consequently, all three problems belong to NP for every fixed undirected
$S$ and every fixed finite directed $S$: a certificate is an ordering of
$V(G)$, checked in $G_S$. Below, NP-hardness therefore suffices to establish
NP-completeness in these cases. For infinite directed sets, constructing
$G_S$ need not be decidable.

\section{Undirected Graphs: NP-Completeness}
\label{sec:undir-hard}

We prove NP-completeness for finite odd prefixes and many-one
NP-completeness for $\{2,4\}$, settling the outstanding undirected
hardness cases of~\cite{ALM}.

\subsection{Finite odd prefixes}

We settle the sole remaining family in the cycle classification
of~\cite{ALM}: every finite nonempty odd prefix gives NP-completeness
for cycles and both path variants.

\begin{theorem}[Finite odd prefixes]\label{thm:odd-hard}
For every fixed finite odd $k\geq1$, all three $O_k$-Hamiltonian problems
are NP-complete. The cycle problem remains NP-complete on connected
bipartite graphs of maximum degree at most $2(k+1)$; both path variants
remain NP-complete on connected bipartite graphs of maximum degree at
most $2(k+1)+1$.
\end{theorem}

The reduction forces each vertex gadget to be traversed consecutively.
The following elementary count will also be used in the directed reduction.

A \emph{run} of a vertex set $Q$ on a path or cycle $W$ is a maximal
consecutive segment of $W$ contained in $Q$. If $W$ is directed, its runs are the connected components of the
subgraph induced by $Q\cap V(W)$ in the underlying undirected graph
of $W$; all traversals must still respect the arc directions.

A block is a designated set of vertices of the reduction; it need not
be consecutive in an arbitrary ordering. Its runs are the consecutive
pieces in which the ordering visits it (Figure~\ref{fig:block-runs}).

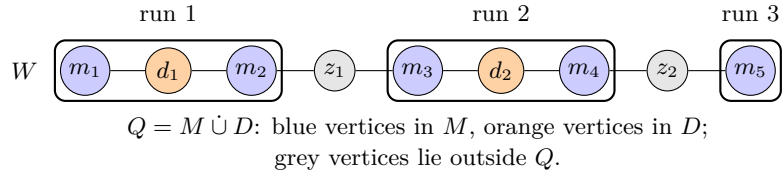
\begin{figure}[ht]
\centering
\begin{tikzpicture}[font=\small,
 mvertex/.style={circle,draw,fill=blue!20,inner sep=2pt},
 dvertex/.style={circle,draw,fill=orange!35,inner sep=2pt},
 outside/.style={circle,draw,fill=gray!20,inner sep=2pt}]
\node[mvertex] (v0) at (0,0) {$m_1$};
\node[dvertex] (v1) at (1.1,0) {$d_1$};
\node[mvertex] (v2) at (2.2,0) {$m_2$};
\node[outside] (v3) at (3.3,0) {$z_1$};
\node[mvertex] (v4) at (4.4,0) {$m_3$};
\node[dvertex] (v5) at (5.5,0) {$d_2$};
\node[mvertex] (v6) at (6.6,0) {$m_4$};
\node[outside] (v7) at (7.7,0) {$z_2$};
\node[mvertex] (v8) at (8.8,0) {$m_5$};
\foreach \i/\j in {0/1,1/2,2/3,3/4,4/5,5/6,6/7,7/8}
 \draw (v\i)--(v\j);
\draw[rounded corners,thick] (-.4,-.4) rectangle (2.6,.4);
\draw[rounded corners,thick] (4,-.4) rectangle (7,.4);
\draw[rounded corners,thick] (8.4,-.4) rectangle (9.2,.4);
\node at (1.1,.75) {run 1};
\node at (5.5,.75) {run 2};
\node at (8.8,.75) {run 3};
\node[left] at (-.5,0) {$W$};
\node[align=center] at (4.4,-.95) {$Q=M\mathbin{\dot\cup}D$: blue vertices in $M$, orange vertices in $D$;\\grey vertices lie outside $Q$.};
\end{tikzpicture}
\caption{One block $Q$ visited in three runs by a spanning path $W$.
The block is the fixed set of coloured vertices, whereas each outlined
run is a maximal consecutive segment of $W$ contained in that set.
The reductions use the counting lemma to force a single run per block.}
\label{fig:block-runs}
\end{figure}

\begin{lemma}\label{lem:block-count}
Let $W$ be a Hamiltonian path or cycle in a graph or digraph $H$, and let
$Q=M\mathbin{\dot\cup}D$ be a nonempty proper subset of $V(H)$. Suppose
that $D$ is independent and all neighbours of vertices in $D$ belong to
$M$; in the directed case this includes both in- and out-neighbours.
Let $e_D$ be the number of endpoints of $W$ in $D$ (zero for a cycle),
and let $z$ count the edges or arcs of $W$ with both ends in $M$.
Then the number $r$ of $Q$-runs is
\begin{equation}\label{eq:block-count}
 r=|M|-|D|+e_D-z.
\end{equation}
In particular, if $|D|=|M|-1$ and $e_D=0$, then $Q$ is one run.
If $M$ is also independent, then $z=0$.
\end{lemma}
\begin{proof}
Since $Q$ is proper, the edges of $W$ induced by $Q$ form a linear forest
with $r$ components. Every edge of $W$ incident with $D$ stays in $Q$, so
this forest has exactly $2|D|-e_D+z$ edges. Consequently,
$r=|Q|-(2|D|-e_D+z)$, proving~\eqref{eq:block-count}. In the stated
special case, $1\leq r=1-z\leq1$.
\end{proof}

\begin{proof}
We reduce from Hamiltonian Cycle on grid graphs, and from Hamiltonian
Path with prescribed endpoints on grid graphs; both source problems are
NP-complete~\cite{IPS82}. Let $F$ be the source graph. For cycles, assume
that $F$ is connected, has at least three vertices, and has minimum
degree at least two. For paths, let $s\ne t$ be the prescribed endpoints,
and assume that $F$ is connected, every non-endpoint has degree at least
two, and both endpoints have positive degree. Inputs violating these
necessary conditions are mapped to fixed no-instances. The path on three
vertices is such a cycle no-instance, and the claw is such a path
no-instance, for every $O_k$. Recall that grid graphs are bipartite and
have maximum degree four.

\paragraph*{Construction.}
Put $\lambda=(k+1)/2$. For every source edge $e=uv$, introduce a private
path $P_e$ of length $k$ with terminals $t_{u,e}$ and $t_{v,e}$. For each
endpoint $v$ of $e$, let $B_{v,e}$ consist of the vertices at even distance
from $t_{v,e}$ along $P_e$. Thus $|B_{v,e}|=\lambda$, and
$B_{u,e},B_{v,e}$ partition $V(P_e)$. Define
\[
 M_v=\bigcup_{e\ni v}B_{v,e},\qquad m_v=|M_v|=\lambda\deg_F(v).
\]
Choose $Z=\varnothing$ for the cycle reduction and $Z=\{s,t\}$ for the
path reductions. Add a set $D_v$ of $m_v-1+\mathbf1_{v\in Z}$ vertices,
each adjacent to every terminal $t_{v,e}$ with $e\ni v$. Set
$Q_v=M_v\mathbin{\dot\cup}D_v$, and call the resulting graph $G_Z$.
For prescribed paths choose $\alpha\in D_s$ and $\beta\in D_t$.
All sets $D_v$ are nonempty by the assumptions on the source.

The sets $Q_v$ partition $V(G_Z)$, and the construction has linear size
for fixed $k$. A bipartition of $F$ extends by placing $M_v$ on the side
of $v$ and $D_v$ on the opposite side. The graph is connected.
A terminal has degree $m_v+\mathbf1_{v\in Z}$, a vertex in $D_v$ has
degree $\deg_F(v)$, and other vertices have degree two. These observations
give the asserted degree bounds.

Figure~\ref{fig:odd-reduction} illustrates the replacement of a source
vertex and its incident edges for $k=3$.

\begin{figure}[ht]
\centering
\includegraphics[width=.95\linewidth]{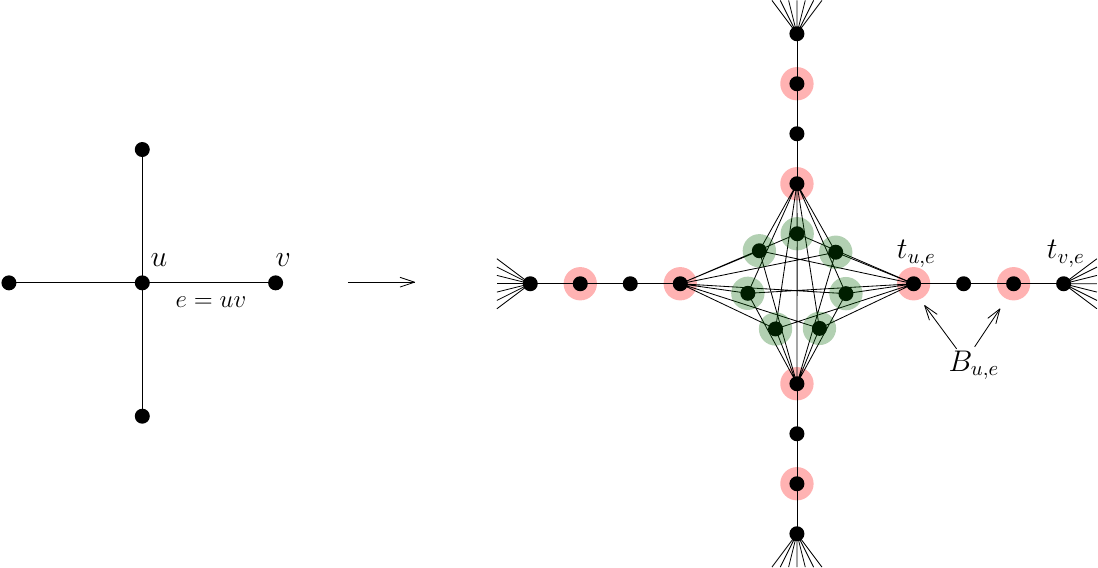}
\caption{The undirected odd-prefix reduction for $k=3$, around a source
vertex $u$ of degree four. Each source edge is replaced by a private
length-three corridor. Red vertices form $M_u$: the even positions
$B_{u,e}$ on each incident corridor. The seven green vertices form $D_u$
in the cycle construction and are each joined to all four terminals
$t_{u,e}$. For the path construction, an endpoint gadget has one
additional vertex in $D_u$.}
\label{fig:odd-reduction}
\end{figure}

Let $H={G_Z}_{O_k}$.
\begin{claim}\label{cl:odd-structure}
For every source vertex $v$ and every $d\in D_v$, $N_H(d)=M_v$, and both
$D_v$ and $M_v$ are independent in $H$. Moreover, an edge of $H$ between
different blocks $Q_u,Q_v$ has its ends in $M_u,M_v$ and satisfies
$uv\in E(F)$. For every source edge $e=uv$, the terminals
$t_{u,e},t_{v,e}$ are adjacent in $H$.
\end{claim}
\begin{proof}
Every vertex of $M_v$ is reachable from $d$ by an odd walk of length at
most $k$: move to the appropriate terminal and then an even number of
steps along its corridor. Conversely, a walk of length at most $k$
starting in $D_v$ cannot traverse an entire incident corridor and reach
its far terminal: that requires at least $k+1$ steps. It remains within
$D_v$ and the incident corridors on the $v$ side of their far terminals.
Among these vertices, precisely those in $M_v$ have the opposite
bipartition colour from $d$. Thus an odd walk of length at most $k$ can
end only in $M_v$. Each of $D_v$ and $M_v$ is monochromatic, and hence
independent in $H$.

It remains to exclude compatibility between $M_u$ and $M_v$ when $uv$
is not a source edge. If $u,v$ lie on the same side of the bipartition
of $F$, their $M$-sets have the same colour and no odd walk joins them.
Suppose therefore that $u,v$ lie on opposite sides and are nonadjacent.
Identify all vertices of each $D_w$ into a single vertex $c_w$. This
preserves distances between corridor vertices: a passage through $c_w$
can be lifted through any vertex of $D_w$. The quotient is the uniform
subdivision of $F$ in which each source edge has length $L=k+2$.

Take $x\in B_{u,ua}$ and $y\in B_{v,vb}$. The distances from $x$ to the
two ends $c_u,c_a$ of its subdivided source edge are at least $1,2$,
respectively; the analogous lower bounds for $y$ and $c_v,c_b$ are
$1,2$. Bipartiteness and the nonadjacency of $u,v$ give
\[
 \dist_F(u,v)\geq3,\quad \dist_F(u,b)\geq2,\quad
 \dist_F(a,v)\geq2,\quad \dist_F(a,b)\geq1.
\]
Any $x$--$y$ path must leave the first subdivided edge and enter the
second through their ends. Its length is therefore at least the minimum
of $2+3L$, $3+2L$, $3+2L$, and $4+L$, which exceeds $k$.
Thus no odd walk of length at most $k$ joins $x$ to $y$.
Finally, $P_e$ itself witnesses the asserted terminal edge.
\end{proof}

\paragraph*{Cycle correctness.}
For $Z=\varnothing$, every block induces $K_{m_v,m_v-1}$ in $H$.
A Hamiltonian cycle of $F$ can be expanded at each $v$ into a spanning
alternating path of $H[Q_v]$ between the two terminals used by that cycle.
The terminal edges from Claim~\ref{cl:odd-structure} join these block
paths into a Hamiltonian cycle of $H$.
Conversely, Lemma~\ref{lem:block-count} makes every $Q_v$ a single run
on any Hamiltonian cycle of $H$. Contracting the runs gives a Hamiltonian
cycle of $F$, because every inter-block edge follows a source edge.

\paragraph*{Path correctness.}
For $Z=\{s,t\}$, the two endpoint blocks induce $K_{m_v,m_v}$, whereas
all other blocks induce $K_{m_v,m_v-1}$. A Hamiltonian $s$--$t$ path of
$F$ expands into a Hamiltonian $\alpha$--$\beta$ path of $H$: use an
alternating path from $\alpha$ to the required terminal in $Q_s$, an
alternating path from the required terminal to $\beta$ in $Q_t$, and
terminal-to-terminal alternating paths in the remaining blocks.

For the reverse direction, let $W$ be \emph{any} Hamiltonian path of $H$.
For each endpoint block, equation~\eqref{eq:block-count} gives
$r=e_{D_v}\geq1$, since its two parts have equal size and $M_v$ is
independent. Hence $W$ has one endpoint in $D_s$ and the other in $D_t$.
Both endpoint blocks are single runs; the other blocks are single runs
by the same lemma. Contracting the runs yields a Hamiltonian $s$--$t$
path of $F$. This proves both path reductions, including the unrestricted
variant.
\end{proof}

\subsection{The length set \texorpdfstring{$\{2,4\}$}{\{2,4\}}}

Here we give polynomial-time many-one reductions, strengthening the
Cook hardness result of~\cite{ALM} for cycles and prescribed
endpoints. The construction also establishes hardness with unrestricted
endpoints.

\begin{theorem}\label{thm:twofour}
All three $\{2,4\}$-Hamiltonian problems are NP-complete on planar graphs
of maximum degree four.
\end{theorem}

We first isolate the two ingredients of the reduction. The first is the
pendant-path construction underlying the classical hardness of
Hamiltonicity in graph squares~\cite{Underground}.

\begin{lemma}\label{lem:pendant-square}
Let $F$ be a graph of order at least three, and obtain $J$ by adding,
for each $v\in V(F)$, new vertices $a_v,b_v$ and edges $va_v,a_vb_v$.
Then $F$ is Hamiltonian if and only if $J^2$ is Hamiltonian. Moreover,
every Hamiltonian cycle of $J^2$ contains every edge $a_vb_v$.
\end{lemma}
\begin{proof}
In $J^2$, the vertex $b_v$ has precisely the neighbours $a_v,v$.
A Hamiltonian cycle must therefore use both incident edges.
It cannot also use $va_v$, because these three edges would form a proper
cycle. The other cycle neighbour of $a_v$ consequently belongs to
$N_F(v)$. There is one such choice for each source vertex, and each original
vertex has only one remaining cycle incidence after its edge to $b_v$.
These choices therefore use all remaining incidences at original vertices. Hence all cycle edges outside the paths
$a_v,b_v,v$ join $a_v$ to a neighbour of $v$ in $F$. Contracting these
paths gives a Hamiltonian cycle of $F$. Conversely, expand every vertex
$v$ of a Hamiltonian cycle of $F$ into $v,b_v,a_v$, oriented in the
cycle order. The connecting edge from $a_v$ to the next original vertex
exists in $J^2$.
\end{proof}

For a graph $J$ and an edge $xy\in E(J)$, construct $B=B(J,xy)$ as
follows. Start with the incidence graph of $J$: replace each vertex $u$
by a node $u'$, each edge $uv$ by a vertex $e_{uv}$, and add
$u'e_{uv},v'e_{uv}$. Write $N=\{u':u\in V(J)\}$. Add seven vertices
$d_1,\ldots,d_7$ and the edges
\begin{equation}\label{eq:incidence-gadget}
 \begin{split}
 &d_1d_3,\ d_2d_3,\ d_3d_4,\ d_4d_5,\ d_5d_6,\ d_5d_7,\ d_6d_7,\\
 &x'd_1,\ y'd_2.
 \end{split}
\end{equation}
This is the incidence gadget of Amarilli, Lombardo, and
Monet~\cite[Proposition~4.2]{ALM}; see Figure~\ref{fig:twofour-gadget}.

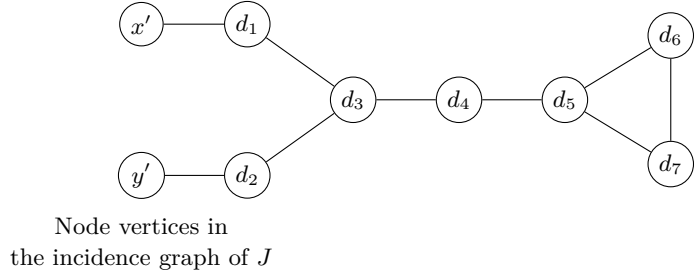
\begin{figure}[ht]
\centering
\begin{tikzpicture}[every node/.style={font=\small},
  vertex/.style={circle,draw,inner sep=2pt,minimum size=5mm}]
\node[vertex] (x) at (0,1) {$x'$};
\node[vertex] (y) at (0,-1) {$y'$};
\node[vertex] (d1) at (1.4,1) {$d_1$};
\node[vertex] (d2) at (1.4,-1) {$d_2$};
\node[vertex] (d3) at (2.8,0) {$d_3$};
\node[vertex] (d4) at (4.2,0) {$d_4$};
\node[vertex] (d5) at (5.6,0) {$d_5$};
\node[vertex] (d6) at (7,0.85) {$d_6$};
\node[vertex] (d7) at (7,-0.85) {$d_7$};
\draw (x)--(d1)--(d3)--(d2)--(y);
\draw (d3)--(d4)--(d5)--(d6)--(d7)--(d5);
\node[align=center] at (0,-1.9) {Node vertices in\\the incidence graph of $J$};
\end{tikzpicture}
\caption{The gadget attached to the incidence graph of $J$. The incidence edges are omitted. The cycle reduction uses one copy. The path reductions use three copies, joined by edges between their $d_6$ vertices.}
\label{fig:twofour-gadget}
\end{figure}

\begin{lemma}[Incidence gadget]\label{lem:incidence-gadget}
Let $H=B_{\{2,4\}}$. Then
\begin{enumerate}[(i)]
\item $H[N]$ is naturally isomorphic to $J^2$;
\item every neighbour of $N$ outside $N$ belongs to $\{d_3,d_5\}$;
\item $N_H(d_5)\cap N=\{x',y'\}$ and
$N_H(d_3)\cap N=\{u':u\in N_J[x]\cup N_J[y]\}$.
\end{enumerate}
If $J^2$ has a Hamiltonian cycle containing $xy$, then $H$ has a
Hamiltonian cycle containing $d_6d_7$.
\end{lemma}
\begin{proof}
For the three structural assertions, an even walk of length at most four
starting at a node cannot reach $d_6$ or $d_7$; each is at distance at
least five from $N$. In the remaining bipartite graph, the vertices on
the node side are exactly $N\cup\{d_3,d_5\}$. Length-two and length-four
walks through the incidence graph encode adjacency in $J^2$. The only
additional node-to-node route of length at most four is
$x',d_1,d_3,d_2,y'$, which adds no new pair because $xy\in E(J)$.
The same inspection gives the two displayed neighbourhoods.

For completeness, we record why the forward construction can retain
$d_6d_7$. Under the forward hypothesis, $J$ is connected. Add a vertex $w$ adjacent to $x,y$ in $J$, and a new leaf $z$
adjacent to $w$, obtaining $J^+$. Its incidence graph is the subgraph of
$B$ on all vertices except $d_6,d_7$, with $w,z$ represented by $d_3,d_5$
and the edges $wx,wy,wz$ by $d_1,d_2,d_4$. The line-graph square lemma
of~\cite[Lemma~4.3]{ALM} states that, in a connected graph with at
least three edges, the square of its line graph has a Hamiltonian cycle
containing any specified pair of distinct edges joined by a length-two
walk in the line graph. Apply it to $wz,wy$, using the intermediate edge
$wx$. Deleting the resulting closing edge gives a path $W$ in $H$ from
$d_4$ to $d_2$ spanning all incidence-edge vertices, including
$d_1,d_2,d_4$.

Write the given node cycle as $(x,y,v_1,\ldots,v_t)$. Then
\[
 (x',d_5,W,d_6,d_7,d_3,y',v_1',\ldots,v_t')
\]
is a Hamiltonian cycle in $H$. Here $W$ includes both its endpoints.
The new transitions are witnessed, respectively, by
\[
 \begin{gathered}
 x'd_1d_3d_4d_5,\quad d_5d_6d_7d_5d_4,\quad
 d_2d_3d_4d_5d_6,\\
 d_6d_5d_7,\quad d_7d_6d_5d_4d_3,\quad d_3d_2y'.
 \end{gathered}
\]
All have length two or four. The second witness deliberately repeats
$d_5$, as the definition permits walks rather than paths.
\end{proof}

\begin{proof}[Proof of Theorem~\ref{thm:twofour}]
We reduce from Hamiltonian Cycle on planar cubic graphs, which is
NP-complete~\cite{GJT76}. Let $F$ be such a source instance and construct
$J$ as in Lemma~\ref{lem:pendant-square}. Choose any $r\in V(F)$ and set
$x=a_r$, $y=b_r$. Construct $B=B(J,xy)$.

The graph $J$ is planar and has maximum degree four. Its incidence graph
is planar, and the additional gadget can be embedded in a face incident
with the pendant incidence path $x',e_{xy},y'$. The attachments increase
the node degrees at $x'$ and $y'$ to three and two, respectively; the
other node degrees are at most four, and all added vertices have degree
at most three. Thus $B$ is planar of maximum degree four. An embedding
can be chosen with $d_6$ on the outer face.

\begin{claim}\label{cl:node-run}
Let $W$ be a Hamiltonian cycle in $B_{\{2,4\}}$, or a Hamiltonian
path with neither endpoint in $N$. Then $N$ is a single run on $W$, and
its two external neighbours are $d_3,d_5$. Consequently, $J^2$ has a
Hamiltonian cycle.
\end{claim}
\begin{proof}
Every $N$-run needs two edges to the interface $I=\{d_3,d_5\}$.
There are at most four such edges. If there were two runs, all cycle or
path incidences of both interface vertices would be used by these runs;
$W[N\cup I]$ would then contain a cycle. This is impossible for a path,
and for a Hamiltonian cycle it would omit other vertices of $B$.
There is therefore one run. Its two external neighbours are distinct,
since attaching both ends to one interface vertex would again form a
proper cycle.

By Lemma~\ref{lem:incidence-gadget}, the node run is a Hamiltonian path
of $J^2$ with one endpoint in $\{x,y\}$ and the other in
$N_J[x]\cup N_J[y]$. The endpoints are distinct and at distance at most
two in $J$, since $xy\in E(J)$. Joining them closes a Hamiltonian cycle
of $J^2$.
\end{proof}

The cycle reduction now follows. A Hamiltonian cycle of $F$ gives one
in $J^2$ containing $xy$, and Lemma~\ref{lem:incidence-gadget} lifts it
to $B$. Conversely, Claim~\ref{cl:node-run} and
Lemma~\ref{lem:pendant-square} recover a Hamiltonian cycle of $F$.

For the path variants, take three disjoint copies $B_1,B_2,B_3$ of $B$
and add $d_6^1d_6^2$ and $d_6^2d_6^3$. Denote the result by $G$ and,
for the prescribed-endpoint variant, choose
$\alpha=d_6^1$ and $\beta=d_7^3$. The graph is planar, and its maximum
degree is still four: only $d_6^2$ receives two extra edges, increasing
its degree from two to four.

Every node of $B_i$ is at distance at least five from $d_6^i$. Hence an
even walk of length at most four starting at a node cannot use an added
edge. The three structural assertions of
Lemma~\ref{lem:incidence-gadget}, and Claim~\ref{cl:node-run}, therefore
remain valid for each node set $N_i$ in $G_{\{2,4\}}$.

If $F$ is Hamiltonian, choose in each copy the Hamiltonian cycle
containing $d_6^id_7^i$ supplied by the lifting lemma. Delete that edge
to obtain a spanning path from $d_6^i$ to $d_7^i$. Concatenate the three
paths using the length-two walks
$d_7^i,d_6^i,d_6^{i+1}$, for $i=1,2$. This gives an
$\alpha$--$\beta$ Hamiltonian path in the compatibility graph.
Conversely, any Hamiltonian path, regardless of its endpoints, has
neither endpoint in at least one of the three sets $N_i$.
Claim~\ref{cl:node-run} for that copy gives a Hamiltonian cycle of $J^2$,
and hence of $F$. All constructions have linear size.
\end{proof}

\section{Undirected Graphs: Algorithms and Classification}
\label{sec:undir-easy}

We give linear-time tests for all remaining undirected cases and
complete the classification in Table~\ref{tab:undirected}.

An allowed walk cannot join different connected components. Thus every
nontrivial spanning instance first requires connectedness. In this
section, graphs have at least three vertices; smaller instances are
handled directly.

\subsection{Even prefixes containing six}

We show that connectedness and non-bipartiteness characterise all three
problems for even prefixes containing six, resolving the
$\{2,4,6\}$ prescribed-endpoint case left open in~\cite{ALM}. The key observation is that the $E_6$-compatibility graph is the cube
of the common-neighbour graph $G_{\{2\}}$, which is connected exactly
when $G$ is connected and non-bipartite. Karaganis's theorem then gives
Hamiltonian-connectedness, yielding all three variants.

\begin{theorem}\label{thm:even-easy}
Suppose that $\operatorname{cl}_2(S)=E_h$ for even $h\geq6$, or for $h=\infty$.
For each of the three undirected $S$-Hamiltonian problems, an input
$G$ is a yes-instance if and only if it is connected and non-bipartite.
These tests take linear time.
\end{theorem}

\begin{lemma}\label{lem:common-neighbour}
For every integer $r\geq1$,
$
 G_{E_{2r}}=G_{\{2\}}^r.
$
Moreover, if $|V(G)|\geq2$, then $G_{\{2\}}$ is connected if and only if $G$
is connected and non-bipartite.
\end{lemma}
\begin{proof}
An even walk of length $2j\leq2r$ in $G$ yields a sequence of every
second vertex. Consecutive distinct vertices in that sequence are
adjacent in $G_{\{2\}}$; deleting consecutive repetitions gives a walk of
length at most $j$. Conversely, each edge of a path in $G_{\{2\}}$ can be
replaced by a length-two walk in $G$. This proves the equality.

If $G$ is disconnected, so is $G_{\{2\}}$. If $G$ is connected and bipartite,
no edge of $G_{\{2\}}$ joins its two nonempty parts. If $G$ is connected and
non-bipartite, every pair of vertices has an even walk: take any walk
between them, and, if necessary, insert a detour through an odd cycle to
change its parity. Taking every second vertex then joins the pair in
$G_{\{2\}}$.
\end{proof}

\begin{proposition}[Karaganis~\cite{Karaganis}]\label{prop:cube}
For every connected graph $G$ with at least three vertices, $G^3$ has
a Hamiltonian path between any prescribed pair of distinct vertices.
It therefore also has an unrestricted Hamiltonian path and a Hamiltonian
cycle: choose adjacent endpoints and close the spanning path with their
edge. Equivalently, $G$ admits all three $\{1,2,3\}$-Hamiltonian variants.
\end{proposition}

\begin{proof}[Proof of Theorem~\ref{thm:even-easy}]
If $G$ is connected and non-bipartite, Lemma~\ref{lem:common-neighbour}
and Proposition~\ref{prop:cube} show that
$
 G_{E_6}=G_{\{2\}}^3
$
is Hamiltonian-connected. The same holds for its supergraph
$G_{S}$. This gives both path variants. Choosing adjacent
endpoints of a Hamiltonian path and closing it gives a Hamiltonian
cycle. Conversely, disconnectedness prevents a spanning ordering, and
in a bipartite graph every even walk stays within one part.
Connectedness and bipartiteness can be tested in $O(|V(G)|+|E(G)|)$ time;
there is no need to construct the possibly dense compatibility graph.
\end{proof}

\subsection{Completing the classification}\label{sec:assembly}

Combining the preceding results with the inherited hardness cases gives
the complete dichotomy in Table~\ref{tab:undirected}: all three problems
are NP-complete in its hard rows and linear-time solvable otherwise.
It suffices to classify $\operatorname{cl}_2(S)$, as explained in the introduction.
The finite odd-prefix cases follow from Theorem~\ref{thm:odd-hard},
and the $\{2,4\}$ cases from Theorem~\ref{thm:twofour}.
For $\{2\}$, the cycle and prescribed-endpoint hardness results are
established in~\cite[Theorems~3.5 and~4.1]{ALM}, and the unrestricted
path result in~\cite[Theorem~E.2]{ALMfull}. The $\{1,2\}$ cycle result
is classical~\cite{Underground}; both path variants are covered
in~\cite[Appendix~E.2]{ALMfull}. These are many-one hardness results.

For exhaustiveness, an odd-only nonempty normalised set is $O_k$ for
finite odd $k$ or $O_\infty$. An even-only nonempty normalised set is
$E_2$, $E_4$, or contains $E_6$. A mixed normalised set either equals
$\{1,2\}$, contains $\{1,2,3\}$, or has the form $\{1\}\cup E_h$ with
$h\geq4$, including infinity. The remaining cases are therefore covered
by Theorem~\ref{thm:even-easy} and the known odd-only and mixed-parity
tractability results~\cite{ALM,ALMfull}. The empty set permits no transition between
distinct vertices and is trivial. All tractable tests are linear-time.
An infinite $S$ has an infinite parity prefix and cannot belong to the
hard rows of Table~\ref{tab:undirected}.

\section{Directed Graphs: Finite Length Sets}
\label{sec:dir-finite}

We answer the finite directed-set question raised in~\cite{ALM}: every
nonempty finite set gives NP-completeness.

\subsection{Hardness on oriented graphs}

Every fixed finite nonempty set $S$ gives NP-completeness for all three
directed problems, even without loops or antiparallel arcs.

\begin{theorem}\label{thm:directed-finite}
For every fixed finite nonempty $S\subseteq\Npos$, all three directed
$S$-Hamiltonian problems are NP-complete on oriented graphs.
\end{theorem}
\begin{proof}
We first specify the source problems. Directed Hamiltonian Cycle is
NP-complete~\cite{Karp1972}, even after restricting the source to oriented
graphs: replace each vertex $v$ by $v^-\to v^0\to v^+$ and each arc
$uv$ by $u^+v^-$. Every Hamiltonian cycle traverses each triple
consecutively, so contraction and expansion prove equivalence.
For paths, split a vertex $r$ of an oriented cycle instance into a source
$s$ and a sink $z$: replace $rv$ by $sv$ and $ur$ by $uz$. Cutting and
closing at $r$ proves NP-hardness of Hamiltonian $s$--$z$ Path on
oriented graphs satisfying $d^-(s)=d^+(z)=0$. Every Hamiltonian path of
such a graph necessarily starts at $s$ and ends at $z$.

Put $t=\max S$. The case $t=1$ is settled by these source reductions,
so assume $t\geq2$. Let $F$ be the appropriate oriented source instance.

\paragraph*{Construction.}
For each $v\in V(F)$, introduce $t$ layers
$L_{v,i}=\{x_{v,i}^0,x_{v,i}^1,x_{v,i}^2\}$, for $0\leq i<t$.
Orient each layer as a directed triangle, and add all arcs from every
layer to the next. Put $M_v=\bigcup_i L_{v,i}$, so $|M_v|=3t$.
Add a set $D_v$ of $3t-1$ vertices, all arcs from $D_v$ to $L_{v,0}$,
and all arcs from $L_{v,t-1}$ to $D_v$. Write
\[
 Q_v=M_v\mathbin{\dot\cup}D_v,\qquad
 a_v=x_{v,0}^0,\quad b_v=x_{v,0}^1,\quad \tau_v=x_{v,t-1}^0.
\]
For each source arc $uv$, add $\tau_u a_v$. There are no other arcs.
Call the result $G$. It is oriented: first and last layers are distinct,
all inter-layer arcs point forward, and no added inter-block arc has its
reverse. Its size is linear in the source size for fixed $S$.
For the prescribed-endpoint variant choose $\alpha=a_s$, $\beta=b_z$;
the unrestricted-endpoint variant uses the same $G$.

Figure~\ref{fig:directed-layers} shows one vertex gadget and an outgoing
arc corresponding to a source arc.

\begin{figure}[ht]
\centering
\begin{tikzpicture}[>=Stealth, font=\small,
 vertex/.style={circle,draw,inner sep=1.7pt},
 bundle/.style={->,thick}]
\foreach \x/\i in {0/0,2.6/1,5.2/2} {
 \node[vertex] (p\i) at (\x,0.65) {};
 \node[vertex] (q\i) at (\x-0.45,-0.15) {};
 \node[vertex] (r\i) at (\x+0.45,-0.15) {};
 \draw[->] (p\i)--(q\i);
 \draw[->] (q\i)--(r\i);
 \draw[->] (r\i)--(p\i);
 \node at (\x,1.2) {$L_{v,\i}$};
}
\node[left] at (q0.west) {$a_v$};
\node[below] at (r0.south) {$b_v$};
\node[below] at (r2.south) {$\tau_v$};
\draw[bundle] (r0)--node[above,font=\scriptsize]{all $9$ arcs}(q1);
\draw[bundle] (r1)--node[above,font=\scriptsize]{all $9$ arcs}(q2);
\node[draw,rounded corners,minimum width=1.2cm,minimum height=1.8cm] at (2.6,-1.95) {};
\node[vertex] (d0) at (2.6,-1.3) {};
\node at (2.6,-1.95) {$\vdots$};
\node[vertex] (d1) at (2.6,-2.6) {};
\node[right] at (3.25,-1.95) {$D_v$: $3t-1$ vertices};
\foreach \d in {d0,d1} {
 \draw[bundle] (\d) to[bend left=12] (q0);
 \draw[bundle] (r2) to[bend left=12] (\d);
}
\node at (.9,-1.4) {\scriptsize all arcs};
\node at (4.35,-1.4) {\scriptsize all arcs};
\node[vertex,label=above:{$a_w$}] (next) at (7.2,-.15) {};
\draw[->] (r2) to[bend left=18] node[above,font=\scriptsize]{$vw\in A(F)$} (next);
\end{tikzpicture}
\caption{The directed vertex gadget, illustrated for $t=3$. Each layer
is a directed triangle. Each thick arrow represents all arcs between the
indicated sets: $D_v$ to the first layer, each layer to the next, and the
last layer to $D_v$; connections to both displayed vertices of $D_v$
are shown, and the same connections apply to every omitted vertex. The thin outgoing arc is $\tau_v a_w$.}
\label{fig:directed-layers}
\end{figure}
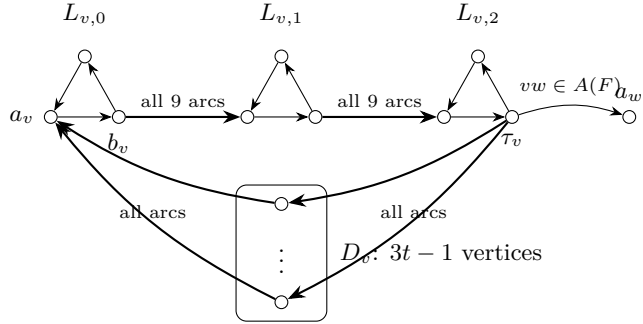

Let $H=G_S$.
\begin{claim}\label{cl:directed-structure}
For every $d\in D_v$, $N_H^+(d)=N_H^-(d)=M_v$. Every inter-block arc
of $H$ goes from $M_u$ to $M_v$ for some $uv\in A(F)$, and every such
source arc gives $b_u a_v\in A(H)$. Finally, for any distinct $a,b\in M_v$,
$H[Q_v]$ has a spanning directed $a$--$b$ path.
\end{claim}
\begin{proof}
Fix $y\in L_{v,j}$. To walk from $d$ to $y$ in exactly $t$ steps,
enter the first layer, take $t-j-1$ steps around its triangle, and
cross $j$ layer boundaries. The entry vertex can be chosen freely, and
complete forward adjacency allows the prescribed final vertex.
For a walk from $y$ to $d$, cross the remaining $t-1-j$ boundaries,
take $j$ steps around the last triangle, and move to $d$.
Both walks have length $t\in S$.

A walk from $D_v$ to a vertex outside $M_v$, or from outside $M_v$ to
$D_v$, needs at least $t+1$ steps: it must pass through the entire
sequence of layers and use an additional incident arc. Thus the stated
neighbourhoods are exact. A walk of length at most $t$ cannot use two
inter-block arcs, since between them it must cross $t-1$ layer boundaries.
Every inter-block compatibility arc therefore follows one source arc.
Conversely, walking from $b_u$ through the layers to $\tau_u$ and then
to $a_v$ takes exactly $t$ steps when $uv\in A(F)$.
Finally, the complete adjacency in both directions between $D_v$ and
$M_v$, with $|D_v|=|M_v|-1$, gives the required spanning alternating
path between any two distinct vertices of $M_v$.
\end{proof}

\paragraph*{Cycle correctness.}
A Hamiltonian cycle of $F$ expands into one of $H$ by replacing each
$v$ with a spanning $a_v$--$b_v$ path in its block and using the arcs
$b_u a_v$ between consecutive blocks. Conversely, in any Hamiltonian
cycle of $H$, Lemma~\ref{lem:block-count} makes every block a single run.
This use of the lemma does not require $M_v$ to be independent.
Contracting the runs gives a Hamiltonian cycle of $F$ by
Claim~\ref{cl:directed-structure}.

\paragraph*{Path correctness.}
A Hamiltonian $s$--$z$ path of $F$ expands in exactly the same way to a
Hamiltonian $\alpha$--$\beta$ path of $H$. Conversely, consider any
Hamiltonian path of $H$. No inter-block arc enters $Q_s$, and none
leaves $Q_z$. Thus the path traverses all of $Q_s$ as its initial segment
and all of $Q_z$ as its final segment. Neither global endpoint lies in
any other block,
so Lemma~\ref{lem:block-count} makes every remaining block a single run.
Contraction gives a Hamiltonian $s$--$z$ path of $F$. This proves the
reverse implication even without prescribed target endpoints.
\end{proof}

\section{Directed Graphs: Infinite Length Sets}
\label{sec:dir-infinite}

Infinite directed sets admit linear-time solvability, NP-completeness,
or undecidability, as the three constructions below show.

\subsection{Cofinite sets: linear-time characterisations}

For cofinite sets, strong connectivity characterises cycles, and a
condition on the condensation characterises paths, giving linear-time
algorithms for all three problems.

A set $S\subseteq\Npos$ is \emph{cofinite} if it contains every sufficiently
large integer. The \emph{condensation} is the acyclic digraph obtained by
contracting each strongly connected component and deleting loops.

\begin{theorem}\label{thm:cofinite}
Fix a cofinite set $S$, and let $G$ be a digraph.
\begin{enumerate}[(i)]
\item For $|V(G)|\geq3$, $G$ has an $S$-Hamiltonian cycle if and only if
it is strongly connected.
\item For $|V(G)|\geq2$, take any topological ordering
$C_1,\ldots,C_k$ of its strongly connected components. An unrestricted
$S$-Hamiltonian path exists if and only if, for every $i<k$, the
condensation has the arc $C_iC_{i+1}$ and, whenever both components consist of
a single vertex without a loop, $1\in S$.
\item A path with distinct prescribed endpoints $a,b$ exists if and
only if the conditions in (ii) hold and $a\in C_1$, $b\in C_k$.
\end{enumerate}
All three problems have linear-time decision algorithms.
\end{theorem}

\begin{lemma}[Cofinite padding]\label{lem:cofinite-padding}
Let $S$ be cofinite. If a directed walk from $x$ to $y$ can pass through
a directed cycle, then a directed $x$--$y$ walk has length in $S$.
\end{lemma}
\begin{proof}
Insert arbitrarily many traversals of that cycle. The resulting walk
lengths are unbounded, so some lie beyond the threshold of $S$.
\end{proof}

\begin{proof}[Proof of Theorem~\ref{thm:cofinite}]
For (i), concatenating the witnessing walks of an $S$-Hamiltonian cycle
shows that every vertex can reach every other. Conversely, in a strongly
connected digraph of order at least two, every ordered pair can be joined
through a directed cycle. By Lemma~\ref{lem:cofinite-padding},
$G_S$ is the complete loopless digraph, proving the assertion.

For the necessity in (ii), let $W$ be an $S$-Hamiltonian path. Once its
ordering leaves a strongly connected component, it cannot return: the
two relevant witnessing walks would put the intervening vertices in
that component. Thus the components occur as consecutive blocks.
Concatenating witnesses shows that each earlier component reaches every
later one. Consequently the condensation has a unique topological
ordering, namely the block order. Consecutive components in this order
must have a direct arc, since no intermediate component can lie between
them. If two consecutive components are single vertices $\{u\}$
and $\{v\}$ without loops, the only possible $u$--$v$ walk is the single arc: a longer
walk would either traverse an intermediate component or create a cycle
in one of the two singleton components. Hence $1\in S$ is necessary.
The first and last components must contain $a,b$ in the prescribed case.

Conversely, suppose the stated conditions hold. List the components in
the given order. Within every component containing at least two vertices, any pair of distinct
vertices is compatible by Lemma~\ref{lem:cofinite-padding}, so its
vertices can be ordered arbitrarily. Between consecutive components,
any chosen last and first vertices can be joined by a directed walk
using the condensation arc. If either component contains at least two vertices or a loop, this walk
can pass through a directed cycle and can be padded to an allowed
length. Otherwise both components are single vertices without loops and the transition is
a single arc, allowed by $1\in S$. Concatenating these orders gives the
required path. In (iii), start the first component at $a$ and end the
last at $b$; when $k=1$, choose both ends of that component accordingly.

Compute the strongly connected components, record their sizes and any
loops, and topologically sort the condensation. A scan of the arcs
marks every pair of consecutive components that is joined by an arc.
Checking these marks, the component sizes and loops, and the endpoints
takes $O(|V(G)|+|A(G)|)$ time. No transitive closure or enumeration of
endpoint pairs is needed.
\end{proof}

\subsection{An Infinite Length Set Giving NP-Completeness}

The fixed infinite set $S_0=\{2\}\cup O_\infty$ gives NP-completeness
for all three directed problems, even on strongly connected bipartite oriented graphs.

\begin{theorem}\label{thm:infinite-hard}
For the fixed infinite set $S_0=\{2\}\cup O_\infty$, all three directed
$S_0$-Hamiltonian problems are NP-complete, even on strongly connected
bipartite oriented graphs.
\end{theorem}
\begin{proof}
To test for an odd-length directed walk from $x$ to $y$, replace each
vertex $v$ by $(v,0),(v,1)$ and each arc $uv$ by the arcs
$(u,0)(v,1)$ and $(u,1)(v,0)$. Every step flips the second coordinate,
so such a walk exists exactly when $(y,1)$ is reachable from $(x,0)$.
Ordinary graph search tests this in polynomial time. Length-two walks
are tested by checking for an intermediate vertex. Together these tests
construct $G_{S_0}$ in polynomial time, proving membership in NP on
arbitrary digraphs.

We reduce from Hamiltonian Path on connected undirected
graphs~\cite{GareyJohnson1979}. Let $F$ be the symmetric orientation of
a connected source graph: each edge gives both opposing arcs. A
\emph{directed path cover} partitions the vertices into directed paths,
where singleton paths are allowed; let $\pc(D)$ be the minimum number
of paths in such a cover of $D$. The source is a yes-instance precisely
when $\pc(F)=1$.

\paragraph*{Construction.}
Fix $r\in V(F)$ and an integer $h\geq2$, to be chosen below. Subdivide
every arc $uv\in A(F)$ by a distinct vertex $y_{uv}$, giving
$u\to y_{uv}\to v$. Add vertices $a,b,c,p,s$ and
$\ell_1,\ldots,\ell_h$, and add exactly the following further arcs:
\[
 \ell_i a,\ b\ell_i\ (1\leq i\leq h),\qquad
 ac,\ cb,\ cp,\ pr,\ rs,\ sc.
\]
Call the resulting digraph $G_h$. It is oriented: opposite source arcs use distinct subdivision
vertices, and none of the added gadget arcs has its reverse.
It is bipartite, with parts
\[
 X_h=V(F)\cup\{c,\ell_1,\ldots,\ell_h\},\qquad
 Y=\{y_{uv}:uv\in A(F)\}\cup\{a,b,p,s\}.
\]
It is strongly connected: the subdivided copy of $F$ is strongly
connected, it reaches $c$ through $r,s$, $c$ reaches it through $p,r$,
and each added leaf vertex is linked to $c$ through $a,b$.
Put $q=|Y|=|A(F)|+4$.

Figure~\ref{fig:infinite-directed} illustrates the construction around
$r$ when $r,u,v$ form a triangle in the source graph.

\begin{figure}[ht]
\centering
\begin{tikzpicture}[>=Stealth,font=\small,scale=.82,
 xvertex/.style={circle,draw,fill=blue!20,inner sep=2pt},
 yvertex/.style={circle,draw,fill=orange!35,inner sep=2pt}]
\node[xvertex] (r) at (0,0) {$r$};
\node[xvertex] (u) at (2,2) {$u$};
\node[xvertex] (v) at (4,0) {$v$};
\node[yvertex] (ru) at (.55,1.35) {$y_{ru}$};
\node[yvertex] (ur) at (1.45,.65) {$y_{ur}$};
\node[yvertex] (uv) at (3.45,1.35) {$y_{uv}$};
\node[yvertex] (vu) at (2.55,.65) {$y_{vu}$};
\node[yvertex] (rv) at (2,-.65) {$y_{rv}$};
\node[yvertex] (vr) at (2,-1.65) {$y_{vr}$};
\foreach \a/\b/\c in {r/ru/u,u/ur/r,u/uv/v,v/vu/u,r/rv/v,v/vr/r} {
 \draw[->] (\a)--(\b); \draw[->] (\b)--(\c);
}
\node[xvertex] (c) at (-3,0) {$c$};
\node[yvertex] (p) at (-1.5,.6) {$p$};
\node[yvertex] (s) at (-1.5,-.6) {$s$};
\node[yvertex] (a) at (-4.2,1.1) {$a$};
\node[yvertex] (b) at (-4.2,-1.1) {$b$};
\node[xvertex] (l1) at (-6,1.1) {$\ell_1$};
\node at (-6,0) {$\vdots$};
\node[xvertex] (lh) at (-6,-1.1) {$\ell_h$};
\foreach \i/\j in {c/p,p/r,r/s,s/c,a/c,c/b,l1/a,lh/a,b/l1,b/lh} {
 \draw[->] (\i)--(\j);
}
\node[xvertex,label=right:{$X_h$}] at (-3.5,-2.4) {};
\node[yvertex,label=right:{$Y$}] at (-1.5,-2.4) {};
\end{tikzpicture}
\caption{The construction for the infinite directed set $S_0$, with a
source triangle on $r,u,v$. Blue vertices belong to $X_h$ and orange
vertices to $Y$. Each source edge gives two opposing arcs, each subdivided
by its own orange vertex. The gadget attaches to the source only at $r$;
every $\ell_i$ has arcs to $a$ and from $b$. Only $\ell_1,\ell_h$ are drawn.}
\label{fig:infinite-directed}
\end{figure}
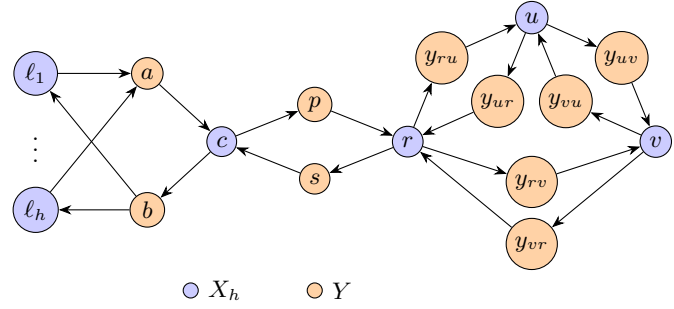

In $(G_h)_{S_0}$, all arcs in both directions between
$X_h$ and $Y$ are present: strong connectivity gives walks between
every ordered pair, and bipartiteness makes every such cross-part walk
odd. Within $X_h$, the only permitted length is two. Inspection of the
construction shows that $J_h=(G_h)_{S_0}[X_h]$ consists precisely of $F$, the
bidirected edges $c\ell_i$ for $1\leq i\leq h$, and the bidirected edge
$cr$.

\begin{claim}\label{cl:path-cover}
$\pc(J_h)=h-1+\pc(F)$.
\end{claim}
\begin{proof}
For the upper bound, use the path $\ell_1,c,\ell_2$, the other $h-2$
leaves as singleton paths, and a minimum path cover of $F$.
For the lower bound, consider the unique cover path containing $c$.
If it avoids $F$, it contains at most two leaves. Together with the
other singleton leaves, this gives at least $h-1$ paths disjoint from
$F$. If it meets $F$, it contains at most one leaf: reaching $F$ uses
one incidence at $c$ through $r$. At least $h-1$ leaves then remain
singleton paths. In either case, a cover path meeting $F$ contains a
single contiguous segment of $F$, because the only attachment is at
$r$ and a path cannot repeat $r$. Restricting the paths that meet $F$
therefore gives a path cover of $F$, and at least $\pc(F)$ such paths
are needed. Adding the disjoint paths proves the lower bound.
\end{proof}

\paragraph*{The three variants.}
Choose the value of $h$ as follows:
\begin{center}
\begin{tabular}{@{}lcc@{}}
\toprule
Target variant & $h$ & Maximum number of $X_h$-runs\\
\midrule
Cycle & $q$ & $q$\\
Path prescribed from $a$ to $b$ & $q-1$ & $q-1$\\
Unrestricted-endpoint path & $q+1$ & $q+1$\\
\bottomrule
\end{tabular}
\end{center}
Indeed, $q$ vertices outside $X_h$ separate at most $q$ cyclic runs,
at most $q-1$ runs of a path with both endpoints in $Y$, and at most
$q+1$ runs of an arbitrary path. These runs form a path cover of $J_h$.
Thus a yes-instance of the chosen target variant implies
$\pc(J_h)\leq h$, and Claim~\ref{cl:path-cover} yields $\pc(F)\leq1$.

Conversely, if $\pc(F)=1$, then $\pc(J_h)=h$. Let
$P_1,\ldots,P_h$ be a cover of this size. Every transition between an
$X_h$-path and a vertex of $Y$ is available in both directions.
For cycles, alternate the $q$ paths with all $q$ vertices of $Y$.
For prescribed paths, start at $a$, alternate the $q-1$ paths with the
$q-2$ vertices of $Y\setminus\{a,b\}$, and end at $b$.
For unrestricted paths, alternate the $q+1$ paths with the $q$ vertices
of $Y$, beginning and ending with a path. This proves the forward
implications. Each construction has linear size, establishing all
three many-one reductions.
\end{proof}

\subsection{Undecidability for a Recursive Length Set}

We construct a fixed recursive infinite set for which all three
directed problems are undecidable, even on strongly connected oriented graphs.

\begin{theorem}\label{thm:undecidable}
There is a fixed recursive infinite set $S^\star\subseteq\Npos$ for
which all three directed $S^\star$-Hamiltonian problems are undecidable,
even on strongly connected oriented graphs.
\end{theorem}

The last construction separates membership in a length set from the
ability to determine whether that set contains a feasible walk length.

\begin{lemma}[Divisibility gadget]\label{lem:divisibility}
Let $L\geq2$, and construct $B_{L,k}$ from an independent set
$D=\{d_1,\ldots,d_k\}$ and vertices $p_1,\ldots,p_{L-1}$ by adding
all arcs $d_i p_1$, all arcs $p_{L-1}d_i$, and the directed path
$p_1\to\cdots\to p_{L-1}$. For any fixed $S$ containing $1$, the
following are equivalent:
\begin{enumerate}[(i)]
\item $S$ contains a positive multiple of $L$;
\item $B_{L,L}$ has an $S$-Hamiltonian cycle;
\item $B_{L,L}$ has an $S$-Hamiltonian path from $d_1$ to $p_{L-1}$;
\item $B_{L,L+1}$ has an unrestricted $S$-Hamiltonian path.
\end{enumerate}
Each gadget is strongly connected and loopless, and has $O(L)$ vertices
and arcs for the two stated values of $k$.
\end{lemma}
\begin{proof}
Every positive-length walk from one vertex of $D$ to another has length
a positive multiple of $L$, since it must traverse the entire $p$-path
between successive visits to $D$. Conversely, every positive multiple
of $L$ occurs between every ordered pair in $D$.

If (i) holds, all transitions between distinct vertices of $D$ are
allowed. Consider the ordering
\[
 d_1,\ldots,d_k,p_1,\ldots,p_{L-1}.
\]
The transitions
inside $D$ use an allowed multiple of $L$, and the remaining transitions
use length one. For $k=L$, the ordering gives the path in (iii), and
closing at $d_1$ gives the cycle in (ii). For $k=L+1$, it gives (iv).

If (i) fails, $D$ is independent in the compatibility digraph.
In a cycle with $k=L$, all $L$ vertices of $D$ need distinct successors
among only $L-1$ vertices outside $D$, which is impossible. The same
count applies to a path ending at $p_{L-1}$, since no vertex of $D$ is
the final vertex. For an unrestricted path with $k=L+1$, at most one
vertex of $D$ can be final, leaving at least $L$ vertices that again
need distinct successors among $L-1$ vertices. This excludes (ii)--(iv).
Strong connectivity, looplessness, and the size bounds are immediate
from the construction.
\end{proof}

\begin{proof}[Proof of Theorem~\ref{thm:undecidable}]
Fix an effective enumeration $T_1,T_2,\ldots$ of Turing machines run on
empty input, and let $q_n$ be the $n$th odd prime. Define
\begin{equation}\label{eq:recursive-set}
 \begin{split}
 S^\star={}&\{1\}\cup\{2^r:r\geq1\} \cup\{2q_n^t:n,t\geq1\text{ and }T_n
                          \text{ halts within }t\text{ steps}\}.
 \end{split}
\end{equation}
This set is infinite because it contains all positive powers of two.
It is recursive: test whether an input is $1$ or a power of two;
otherwise factor it. The only remaining possible form is $2p^t$ with
$p$ an odd prime and $t\geq1$. Compute the index $n$ of $p$ among the
odd primes and simulate $T_n$ for $t$ steps. This is a terminating
membership test; no polynomial-time bound is asserted or needed.

Unique prime factorisation gives
\begin{equation}\label{eq:halting-multiples}
 S^\star\cap q_n\Npos\ne\varnothing
 \quad\Longleftrightarrow\quad T_n\text{ halts}.
\end{equation}
Indeed, among the integers in~\eqref{eq:recursive-set}, only numbers
of the form $2q_n^t$ are divisible by $q_n$, and such a number is included
exactly when a sufficiently long finite simulation witnesses halting.

Given $n$, compute $q_n$ and apply Lemma~\ref{lem:divisibility} with
$L=q_n$. Equation~\eqref{eq:halting-multiples} reduces the halting problem
to the cycle and prescribed-endpoint problems on $B_{q_n,q_n}$, and to
the unrestricted-endpoint problem on $B_{q_n,q_n+1}$. Since $L=q_n\geq3$, the vertices $p_1$ and $p_{L-1}$ are distinct,
and the gadgets contain no antiparallel arcs. These are therefore
computable many-one reductions to finite strongly connected
oriented graphs.
The undecidability of halting~\cite{Turing1936} proves the theorem.
\end{proof}

\section*{Declaration of AI use}
A preliminary version of Theorem~\ref{thm:directed-finite} was obtained
manually using a construction with loops. GPT~6 Pro was used to extend
this result to oriented graphs and to obtain all other new results
in this paper. The authors take full responsibility for the content
of the paper.

\bibliography{bibliography}

@article{Gould2003,
  author  = {Gould, Ronald J.},
  title   = {Advances on the {Hamiltonian} Problem---A Survey},
  journal = {Graphs and Combinatorics},
  volume  = {19},
  number  = {1},
  pages   = {7--52},
  year    = {2003},
  doi     = {10.1007/s00373-002-0492-x}
}

@book{ApplegateEtAl2006,
  author    = {Applegate, David L. and Bixby, Robert E. and
               Chv{\'a}tal, Va{\v{s}}ek and Cook, William J.},
  title     = {The Traveling Salesman Problem: A Computational Study},
  series    = {Princeton Series in Applied Mathematics},
  publisher = {Princeton University Press},
  year      = {2006},
  url       = {https://www.jstor.org/stable/j.ctt7s8xg}
}

@inproceedings{ALM,
  author    = {Amarilli, Antoine and Lombardo, Arthur and Monet, Mika{\"e}l},
  title     = {The {$S$}-Hamiltonian Cycle Problem},
  booktitle = {52nd International Workshop on Graph-Theoretic Concepts in
               Computer Science (WG 2026)},
  editor    = {Goedgebeur, Jan and Rz\k{a}\.{z}ewski, Pawe{\l}},
  series    = {Leibniz International Proceedings in Informatics (LIPIcs)},
  volume    = {376},
  pages     = {5:1--5:15},
  publisher = {Schloss Dagstuhl -- Leibniz-Zentrum f{\"u}r Informatik},
  address   = {Dagstuhl, Germany},
  year      = {2026},
  doi       = {10.4230/LIPIcs.WG.2026.5}
}

@misc{ALMfull,
  author        = {Antoine Amarilli and Arthur Lombardo and Mika{\"e}l Monet},
  title         = {The {$S$}-Hamiltonian Cycle Problem},
  year          = {2026},
  month         = may,
  howpublished  = {arXiv:2602.16532v2},
  eprint        = {2602.16532},
  archivePrefix = {arXiv},
  primaryClass  = {cs.DS},
  doi           = {10.48550/arXiv.2602.16532},
  url           = {https://arxiv.org/abs/2602.16532v2},
  note          = {Full version, version 2}
}

@book{GareyJohnson1979,
  author    = {Garey, Michael R. and Johnson, David S.},
  title     = {Computers and Intractability: A Guide to the Theory of
               {NP}-Completeness},
  publisher = {W. H. Freeman},
  address   = {San Francisco},
  year      = {1979},
  isbn      = {978-0-7167-1044-8}
}

@article{GJT76,
  author  = {Garey, Michael R. and Johnson, David S. and Tarjan, Robert Endre},
  title   = {The Planar {Hamiltonian} Circuit Problem is {NP}-Complete},
  journal = {SIAM Journal on Computing},
  volume  = {5},
  number  = {4},
  pages   = {704--714},
  year    = {1976},
  doi     = {10.1137/0205049}
}

@article{IPS82,
  author  = {Itai, Alon and Papadimitriou, Christos H. and
             Szwarcfiter, Jayme Luiz},
  title   = {{Hamilton} Paths in Grid Graphs},
  journal = {SIAM Journal on Computing},
  volume  = {11},
  number  = {4},
  pages   = {676--686},
  year    = {1982},
  doi     = {10.1137/0211056}
}

@article{Karaganis,
  author  = {Karaganis, Jerome J.},
  title   = {On the Cube of a Graph},
  journal = {Canadian Mathematical Bulletin},
  volume  = {11},
  number  = {2},
  pages   = {295--296},
  year    = {1968},
  doi     = {10.4153/CMB-1968-037-0}
}

@incollection{Karp1972,
  author    = {Karp, Richard M.},
  title     = {Reducibility among Combinatorial Problems},
  booktitle = {Complexity of Computer Computations},
  editor    = {Miller, Raymond E. and Thatcher, James W. and Bohlinger, Jean D.},
  publisher = {Plenum Press},
  address   = {New York},
  pages     = {85--103},
  year      = {1972},
  doi       = {10.1007/978-1-4684-2001-2_9}
}

@article{Turing1936,
  author  = {Turing, Alan M.},
  title   = {On Computable Numbers, with an Application to the
             {Entscheidungsproblem}},
  journal = {Proceedings of the London Mathematical Society},
  volume  = {s2-42},
  number  = {1},
  pages   = {230--265},
  year    = {1937},
  doi     = {10.1112/plms/s2-42.1.230}
}

@article{Underground,
  author  = {Underground, P.},
  title   = {On graphs with {Hamiltonian} squares},
  journal = {Discrete Mathematics},
  volume  = {21},
  number  = {3},
  pages   = {323},
  year    = {1978},
  doi     = {10.1016/0012-365X(78)90164-4}
}
\end{document}